\documentclass[prd,reprint, nofootinbib,amsmath,amssymb, aps, floatfix]{revtex4-1}
\usepackage{dcolumn}% Align table columns on decimal point

\usepackage{bm}% bold math
\usepackage{amssymb,amsmath,amsthm,graphicx,ulem}
\usepackage[linktoc=page]{hyperref}

\usepackage{graphicx,subfigure}
\usepackage{epsfig}
\usepackage{amsmath}
\usepackage{amsfonts}
\usepackage{amssymb}
\usepackage[usenames]{color}
\usepackage{enumerate}
\usepackage{setspace}

\usepackage[letterpaper,left=2.cm,right=2.cm,top=2.5cm,bottom=2.5cm]{geometry}
\usepackage[T1]{fontenc}
\newtheorem{thm}{Theorem}[section]
\theoremstyle{lemma}
\theoremstyle{conjecture}
\theoremstyle{definition}
\theoremstyle{convention}
\theoremstyle{corollary}
\newtheorem{cor}[thm]{Corollary} 
\newtheorem{lem}[thm]{Lemma}
\newtheorem{defn}[thm]{Definition}
\newtheorem{rem}[thm]{Remark}
\newtheorem{conj}[thm]{Conjecture}

\begin{document}
\normalem

\title{Generalized Second Law for Cosmology}
\author{Raphael Bousso}%
 \email{bousso@lbl.gov}
\affiliation{ Center for Theoretical Physics and Department of Physics\\
University of California, Berkeley, CA 94720, USA 
}%
\affiliation{Lawrence Berkeley National Laboratory, Berkeley, CA 94720, USA}

\author{Netta Engelhardt}
\email{engeln@physics.ucsb.edu}
\affiliation{Department of Physics, University of California, Santa Barbara, CA 93106, USA 
}%
\bibliographystyle{utcaps}

\begin{abstract}
We conjecture a novel Generalized Second Law that can be applied in cosmology, regardless of whether an event horizon is present: the generalized entropy increases monotonically outside of certain hypersurfaces we call past Q-screens. A past Q-screen is foliated by surfaces whose generalized entropy (sum of area and entanglement entropy) is stationary along one future null direction and increasing along the other.  We prove that our Generalized Second Law holds in spacetimes obeying the Quantum Focussing Conjecture. An analogous law applies to future Q-screens, which appear inside evaporating black holes and in collapsing regions.\\[5ex]\begin{center} {\normalsize \em Dedicated to the memory of Jacob Bekenstein} \end{center}
\end{abstract}

%\pacs{}
\maketitle

%%%%%%%%%%%%%%%%%%%%%%%%%%%%%%%%%%%%%%%%%%%%%%%%%%%%%%%%%%%%%%%%%%%%%%%%%%%
%%%%%%%%%%%%%%%%%%%%%%%%%%%%%%%%%%%%%%%%%%%%%%%%%%%%%%%%%%%%%%%%%%%%%%%%%%%

\tableofcontents

\begin{spacing}{1.1}
\section{Introduction}

The thermodynamics of gravitating systems is a fundamental link between quantum phenomena and gravity. This connection is manifest in various contexts (\textit{e.g.}~\cite{Bek72,Haw75,CEB1, Mal01}), suggesting that it is borne of an underlying principle of full quantum gravity. Hawking's classical area theorem~\cite{Haw71}, an early indication of this connection, states that the area of a black hole event horizon cannot decrease. 

Hawking's theorem holds in spacetimes obeying the null curvature condition, $R_{ab} k^a k^b \geq 0$ for any null vector $k^a$. This will be the case if the Einstein equations are obeyed with a stress tensor satisfying the Null Energy Condition (NEC), 
\begin{equation}
T_{ab} k^a k^b \geq 0~.
\end{equation}
The NEC is satisfied by ordinary classical matter, but it is violated by valid quantum states (e.g., in the Standard Model). In particular, the NEC fails in a neighborhood of a black hole horizon when Hawking radiation is emitted. Indeed, the area of the event horizon of an evaporating black hole decreases, violating the Hawking area law. 

Bekenstein proposed that the area of an event horizon should be interpreted as an entropy: $S_\mathrm{BH} \equiv A_{EH}/4G\hbar$. He further proposed the Generalized Second Law of thermodynamics (GSL)~\cite{Bek72,Bek73,Bek74},
\begin{equation}
d S_\mathrm{gen} \geq 0~,
\label{eq-gslintro}
\end{equation}
in which the Bekenstein-Hawking entropy $S_{BH}$ of black holes is properly included in the total entropy budget:
\begin{equation}
S_\mathrm{gen} \equiv S_\mathrm{out} + \frac{A_{EH}}{4G\hbar}~.
\label{eq-sgn}
\end{equation}
The quantity $S_{\mathrm{out}}$ is the von Neumann entropy of the matter outside the black hole. With this generalization, when matter disappears behind an event horizon, an increase in horizon area can compensate for the loss of matter entropy. Thus, the GSL can prevent what would otherwise be a violation of the (ordinary) second law to an external observer.

In light of the breakdown of the area theorem during black hole evaporation, Bekenstein's GSL can also be viewed as the semiclassical extension of Hawking's area theorem. The GSL remains valid even when the NEC is violated and the event horizon shrinks. This is because the exterior entropy $S_\mathrm{out}$ is increased by the Hawking radiation, more than compensating for the area loss~\cite{Pag76}.\footnote{With unitary evolution, entropy cannot increase except under coarse-graining.  For recently formed black holes, Eq.~(\ref{eq-sgn}) suffices since the area term implicitly entails coarse-graining. At late times, the radiation will be the larger system; if the evaporation process is unitary then this era is not strictly in the semiclassical regime~\cite{Pag93}. Nevertheless, Eq.~(\ref{eq-gslintro}) continues to hold under coarse-graining, in the same sense in which the ordinary second law holds in the evaporation of an ordinary matter object in a pure state.} Proofs of the GSL exist for nontrivial limiting regimes; see~\cite{Wal09} for a review and~\cite{Wal10,Wal11} for recent work.

The area theorem and the GSL are associated with the event horizon, or more generally with causal horizons such as the Rindler horizon of an accelerated observer. This limits their applicability: not all observers accelerate eternally, and not all spacetimes have an event horizon. In particular, cosmological solutions (except for asymptotically de~Sitter universes) do not have an event horizon. 

No general formulation of a second law of thermodynamics has been known in cosmology. In the absence of asymptotic regions, the entire spacetime is highly dynamical, so matter and entropy can freely move around. There do not exist natural divisions into subsystems whose entropy could be tracked. In a spatially homogeneous universe, one could consider the comoving entropy density, but this is an approximate notion. It has no fundamental status, and its definition breaks down as density perturbations grow strong. 

Even if a causal horizon can be defined, its location is ``teleological'': it depends on the arbitrarily distant future. Thus, the very notion of a black hole requires a certain asymptotic structure of spacetime and is not rigorously defined in cosmology. 

Since the above limitations stem from the event horizon's dependence on the asymptotic boundary, it would be desirable to identify a more local alternative to the event horizon: a geometric object that satisfies some area law or GSL, but which is rigorously defined in general spacetimes without reference to an asymptotic region. 

Holographic screens~\cite{CEB2} are quasi-locally defined, and can be constructed in general cosmological solutions. (See~\cite{Hay93,AshKri02} for pioneering work on a more restrictive class of quasi-local horizons.) Moreover, we recently proved that future (or past) holographic screens obey an area theorem~\cite{BouEng15a,BouEng15b}, assuming the NEC holds. Thus, they satisfy the criteria outlined above: for a black hole, the holographic screen shares key properties with the event horizon. But holographic screens require no asymptotic structure and exist in more general settings. Moreover, our proof demonstrated that a holographic screen is uniquely associated with each choice of null foliation; it can be constructed simply by maximizing the area on each null slice.

In this paper, we turn to the semiclassical case, where the NEC need not hold. Our area theorem, like Hawking's, may fail in this case. We consider the question of whether holographic screens satisfy a Generalized Second Law instead. We define the notion of a {\em Q-screen}, a quantum corrected holographic screen. A Q-screen $H$ is a hypersurface foliated by spatial surfaces $\sigma(r)$.  Each $\sigma(r)$ extremizes the generalized entropy, under variations along a null hypersurface $N(r)$ orthogonal to it. A Q-screen is called {\em past} (or {\em future}) if the generalized entropy increases (or decreases) along the opposite null direction orthogonal to $\sigma$.

We conjecture that any past or future Q-screen satisfies a novel Generalized Second Law: the generalized entropy along the Q-screen increases monotonically. Assuming the Quantum Focussing Conjecture (QFC)~\cite{BouFis15a} (a quantum extension of the Bousso bound), we show that a Q-screen is again uniquely associated with any null foliation of the spacetime; this in turn implies that our novel GSL holds.

\paragraph*{Outline} In Sec.~\ref{sec-conjecture}, we follow Bekenstein's step of replacing area with generalized entropy in appropriate definitions and statements. First, this modifies the notion of holographic screen by a quantum correction, leading to our definition of Q-screens. Second, the statement that the area increases along a past holographic screen becomes our conjecture of a novel Generalized Second Law: the generalized entropy increases monotonically outside of a past Q-screen. We believe that this is the first thermodynamic law that applies in arbitrary spacetimes, and in particular in cosmology. We also conjecture that the generalized entropy outside of future Q-screens increases monotonically (but towards the past).

In Sec.~\ref{sec-examples}, we consider some examples. In Sec.~\ref{sec-bh}, we show that the classical area law for holographic screens fails for an evaporating black hole. We construct a Q-screen, and we verify that it satisfies the new GSL, due to the contribution of the Hawking radiation to $S_\mathrm{out}$. In Sec.~\ref{sec-cosmo} we construct a Q-screen in cosmology. We find that the new GSL is satisfied, because the area increase greatly dominates over any changes in entropy (much as Bekenstein's GSL tends to be comfortably satisfied when matter enters a black hole). 

In Sec.~\ref{sec-proof}, we show that our GSL follows from the recently proposed Quantum Focussing Conjecture (QFC)~\cite{BouFis15a}. The QFC itself has not been proven generally, but no counterexamples are known. Moreover, the QFC is plausible in that it unifies several nontrivial statements for which proofs do exist, such as Bekenstein's GSL in certain regimes~\cite{Wal10,Wal11}, the Bousso bound in the hydrodynamic regime~\cite{FMW,BouFla03,StrTho03}, and the Quantum Null Energy Condition~\cite{BouFis15b}. 

Throughout this paper, we work in (3+1)-dimensions; the generalization to higher dimensions is trivial. A {\em hypersurface} has codimension 1 in the spacetime; by {\em surface} we always mean a codimension 2 spatial surface (except in the term {\em Cauchy surface}, which as usual refers to an achronal hypersurface).

\paragraph*{Discussion} 

Our GSL extends a central notion of thermodynamics to general spacetimes, particularly to cosmological settings. It adds another link to what appears to be a rich interplay between geometry, energy, and quantum information (e.g.\ \cite{Bek81,Tho93,Sus95,CEB1,CEB2,FMW,Cas08,Wal10QST,%
  Wal11,BCFM1,BCFM2,BouFis15a,BouFis15b}). This web of relations must originate with the emergence of classical spacetime from an underlying quantum gravity theory---an expectation largely borne out in the main example we have of such a theory, the AdS/CFT correspondence~\cite{Mal97,SusWit98,Mal01,RyuTak06,HubRan07}).  A broadened understanding of the second law may yield insights on how to construct a quantum gravity theory for more realistic spacetimes.

We were led to our conjecture as a natural generalization of the area law for holographic screens~\cite{BouEng15a}, which we recently identified and proved. But as far as we can see, neither our area law nor our GSL is ``necessary'' in the same sense as their analogues for event horizons were: Hawking's area law (in hindsight) encodes the second law for purely gravitational systems, and Bekenstein's GSL preserves the second law when both matter and black holes are present. By contrast, it is not clear which ``ordinary second law'' (or other well-established principle) would be violated if we failed to consider the generalized entropy outside Q-screens. 

There may not be a good answer to this question, short of a full quantum gravity theory. This is similar to the difference between entropy bounds that apply strictly to black hole horizons (and which are thus suggested by the GSL), and the more general entropy bounds that appear to hold far more broadly~\cite{CEB1,CEB2,BouFis15a}, for no reason discernible in an existing framework. 

The Bousso bound does single out holographic screens, drawing attention to these particular hypersurfaces and leading us to their further study. Similarly, the QFC singles out Q-screens as preferred hypersurfaces in the spacetime. Let us discuss this in more detail. 

\paragraph*{Relation to the Bousso bound and the QFC} The Bousso bound provides a notion of entropy associated to the area of an arbitrary surface $\sigma$. The area of $\sigma$ yields a bound on the entropy of its lightsheets, null surfaces generated by nonexpanding light-rays orthogonal to $\sigma$:
\begin{equation}
S_{\mathrm{lightsheet}}\leq \frac{A_{\sigma}}{4 G\hbar}~.
\end{equation}
Any surface has four null congruences emanating from it (future-outwards, future-inwards, past-outwards, and past-inwards). At least two of these must be lightsheets (see e.g.\ Fig. 1 of~\cite{CEB1}). 

For example, the event horizon of a classical black hole to the past of any cross-section $\sigma$ is a lightsheet of $\sigma$. In this special case, the Bousso bound implies that the area of the horizon more than compensates for the matter entropy that entered the black hole prior to $\sigma$, consistent with the GSL.

The Bousso bound in particular distinguishes {\em marginal surfaces}: $\sigma$ is marginal if one of its orthogonal null congruences has locally vanishing expansion $\theta=0$ everywhere. In other words, $\sigma$ locally extremizes the area on a null hypersurface $N$ orthogonal to it. Therefore one can regard $N$ as the union of two valid lightsheets. By the Bousso bound, the area of $\sigma$ bounds the entropy on an entire null slice $N$.

Given a null foliation, one can find the surface $\sigma(r)$ of maximal area on each null slice $N(r)$. The union of the $\sigma(r)$ forms a hypersurface $H=\bigcup \sigma(r)$ (not necessarily of definite signature), termed a holographic screen hypersurface in~\cite{CEB2}. The Bousso bound implies that at every time $r$, all the information about the null slice $N(r)$ can be stored on the surface $\sigma(r)$, at a density of no more than one bit per Planck area. This construction makes concrete earlier speculations that the world is like a hologram~\cite{Tho93,Sus95,FisSus98}. Our recent area theorem applies to holographic screen hypersurfaces $H$ that are subject to an additional refinement, analogous to the distinction between past and future event horizons~\cite{BouEng15a,BouEng15b}.

The QFC is a quantum generalization of the Bousso bound, which reduces to it when matter systems are well-isolated on the lightsheet, or when the entropy can be treated in a hydrodynamic approximation. It is based on a quantum generalization of the notion of expansion, defined using the generalized entropy rather than the area of surfaces. All relevant definitions will be presented in the main text. 

Under the QFC, a surface $\sigma$ that maximizes the generalized entropy outside a null slice $N$ is a preferred cross-section of $N$. The union of such surfaces $\sigma(r)$ over a null foliation $N(r)$ defines a quantum-corrected holographic screen $H=\bigcup \sigma(r)$, which we call Q-screen. Our GSL conjecture states that the entropy outside any past or future Q-screen is monotonic in $r$.

\section{Generalized Second Law for Q-Screens} \label{sec-conjecture}

In this section, we state our conjecture. We will begin by reviewing two important quantities that can be associated with a surface, given minimal additional structure: the {\em generalized entropy}, $S_\mathrm{gen}$, and the {\em quantum expansion}, $\Theta$. A {\em quantum marginal surface}, $\sigma$, has vanishing quantum expansion in one null direction. If the quantum expansion in the other null direction has definite sign, then $\sigma$ is said to be {\em marginally quantum trapped} or {\em antitrapped}. 

Quantum marginal surfaces combine to form a {\em Q-screen}, a 3-dimensional hypersurface that need not have definite signature. A Q-screen is called {\em future (past)} if its constituent marginal surfaces are in addition marginally quantum (anti)trapped. We conjecture that the generalized entropy outside a future or past Q-screen always increases.

\subsection{Generalized Entropy and Quantum Expansion}
\label{sec-genen}

We will begin by extending~\cite{Wal10QST,MyePou13,EngWal14,BouFis15a} the notion of generalized entropy to surfaces that need not lie on an event horizon. We consider a globally hyperbolic spacetime (which may be extendible to one that is not: e.g., a domain of dependence in asymptotically Anti-de Sitter space). Let $\sigma$ be a spacelike surface that splits a Cauchy surface $\Sigma$ into two portions; see Fig.~\ref{fig-one}. We may choose either side of $\sigma$ arbitrarily and refer to this portion of $\Sigma$ as $\Sigma_\mathrm{out}$.
\begin{figure}[ht]

\includegraphics[width=.5 \textwidth]{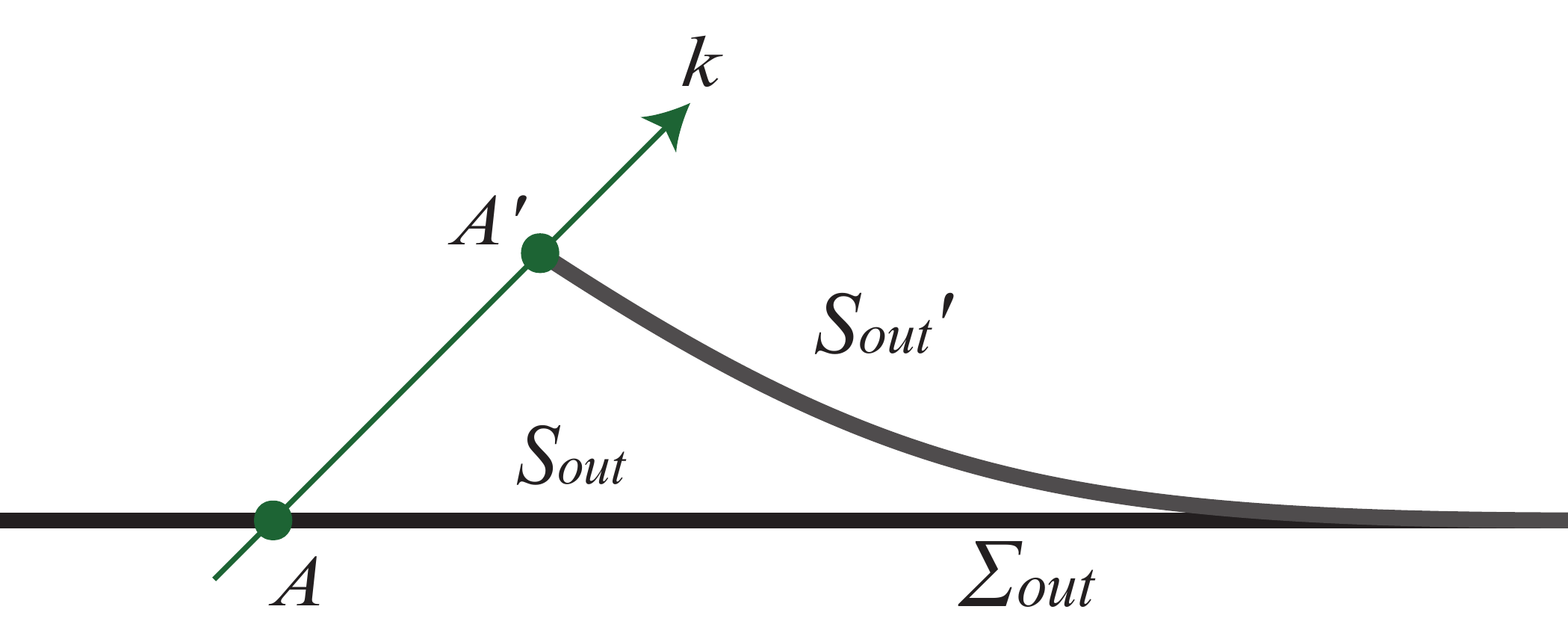}
\caption{The generalized entropy $S_\mathrm{gen}$ is the area $A$ (in Planck units) of a surface that splits a Cauchy surface, plus the von Neumann entropy $S_\mathrm{out}$ of the quantum fields on one side $\Sigma_\mathrm{out}$.  The quantum expansion $\Theta_k$ is the rate at which $S_\mathrm{gen}$ changes as the splitting surface is varied in the orthogonal null direction $k^a$.}\label{fig-one}
\end{figure}

\begin{defn}
The {\em generalized entropy} is the area of $\sigma$ (in Planck units), plus the von Neumann entropy of the quantum state on $\Sigma_\mathrm{out}$:
\begin{equation}
S_\mathrm{gen} \equiv S_\mathrm{out} + \frac{A}{4G\hbar} +\mathrm{counterterms}~,
\label{eq-sgendef}
\end{equation}
where
\begin{equation}
S_\mathrm{out} = - \mathrm{tr}\, \rho_{\rm out} \log \rho_{\rm out}~,
\label{eq-soutdef}
\end{equation}
and the reduced density operator $\rho_{\rm out}$ is the restriction of the global quantum state $\rho$ to $\Sigma_\mathrm{out}$:
\begin{equation}
\rho_{\rm out} = \mathrm{tr}_{\neg \rm out}\, \rho~,
\label{eq-routdef}
\end{equation}
where the trace is taken over the field theory degrees of freedom in the complement of $\Sigma_\mathrm{out}$ on $\Sigma$.\end{defn}

The von Neumann entropy $S_\mathrm{out}$ diverges in regular global states. For example, the entanglement of short-distance degrees of freedom across $\sigma$ in the vacuum contributes a divergence proportional to the area of $\sigma$ in units of the short distance cutoff~\cite{Sorkin83, BKLS86, Srednicki93}. There is compelling evidence that this divergence is cancelled by a renormalization of Newton's constant in the area term~\cite{Susskind:1994sm,Jacobson:1994iw,Solodukhin:1994yz}. From this viewpoint, the Bekenstein-Hawking entropy is the first in a series of counterterms. Subleading divergences are cancelled by other geometric counterterms~\cite{FS94,DLM95} which can be thought of as higher-curvature corrections to the area term \cite{Wald93, JKM93, IW94, IW95}. A review of these arguments and further references can be found in the Appendix of Ref.~\cite{BouFis15a}. Below we will assume that $S_\mathrm{gen}$ is indeed finite and independent of the UV cutoff.

Generalized entropy was originally defined for the case where $\sigma$ is a cross-section of a black hole event horizon~\cite{Bek72}. In this case one takes $\Sigma_\mathrm{out}$ to be the exterior of the black hole. The GSL as formulated by Bekenstein is the statement that $S_\mathrm{gen}$ cannot decrease under forward time evolution of $\Sigma$ along the event horizon.  

We now turn to defining the quantum expansion. There are four families of light-rays emanating orthogonally from the surface $\sigma$: future-outward, future-inward, past-outward, and past-inward. Consider one one of these four families, with tangent vector $k^a$; and consider one of its light-rays, emanating from the point $y_1\in\sigma$. Then deform $\sigma$ in a neighborhood of $y_1$, with infinitesimal area $\cal A$, by an infinitesimal affine distance $\lambda$ along the light-ray; see Fig.~\ref{fig-one}. This yields a new surface with generalized entropy $S_\mathrm{gen}'$ (computed with respect to the same side as $\Sigma_\mathrm{out}$). 
\begin{defn}
The {\em quantum expansion} is given by
\begin{equation}
\Theta_k[\sigma;y_1]\equiv \lim_{{\cal A} \to 0} \left.\frac{4G\hbar}{\cal A} \,
\frac{dS_\mathrm{gen}}{d\lambda}\right|_{y_1}~.
\label{eq-qexpdef}
\end{equation}
In other words, the quantum expansion is the rate of change, per unit area, of the generalized entropy under deformations of $\sigma$ along an orthogonal light-ray. For further details, and an equivalent definition in terms of a functional derivative, see Ref.~\cite{BouFis15a}.\end{defn}

\subsection{Quantum Marginal Surfaces and Q-Screens}

We now require in addition that $\sigma$ be compact and connected. The following definitions follow~\cite{Wal10QST}; they reduce to more familiar classical definitions under the substitution $\Theta\to\theta$. 

\begin{defn} Let $\sigma$ be a compact, connected surface that splits a Cauchy surface into two portions. If one of its orthogonal null congruences, say in the $k^a$ direction, has vanishing quantum expansion everywhere on $\sigma$, we call $\sigma$ a {\em quantum marginal surface}. \end{defn}

\begin{defn} A {\em Q-screen}\footnote{We thank Z.~Fisher for suggesting this term.} $H$ is a smooth hypersurface admitting a foliation by quantum marginal surfaces called {\em leaves}.  \end{defn}

The foliation structure implies that we can think of any screen $H$ as a one-parameter family of marginal surfaces $\sigma(r)$, with the (nonunique) parameter $r$ taking values in an open interval. Moreover, this defines a nowhere vanishing vector field $h^a$ on $H$, which is tangent to $H$ and normal to its leaves. For a given choice of foliation parameter, the normalization of $h$ can be fixed by choosing $h(r) = h^a (dr)_a = 1$, and $h$ can be uniquely decomposed into the null normals $k$ and $l$:
\begin{equation}
h^a = \alpha l^a + \beta k^a~.
\label{eq-hlk}
\end{equation}

It will be convenient to impose a number of weak technical conditions on $H$:
\begin{defn} \label{def-technical}
A Q-screen $H$ is {\em regular} if
\begin{enumerate}[(a)] % (a), (b), (c), ...
\item \label{def-technical1} the {\em quantum generic condition} is met: for any leaf $\sigma$, the quantum expansion $\Theta_k$ at the null geodesic intersecting $\sigma$ at $y_1$ does not continue to vanish when $\sigma$ is infinitesimally deformed along the null generator emanating from $y_2$ along the $k^a$ direction, for any $y_2\in\sigma$ (including $y_2=y_1$).
\item \label{def-technical2} the {\em second generic condition} holds: let $H_+$, $H_-$, $H_0$ be the set of points in $H$ with, respectively, $\alpha>0$, $\alpha<0$, and $\alpha=0$. Then $H_0= \dot H_- = \dot H_+$.
\item \label{def-technical3} every inextendible portion $H_i\subset H$ with definite sign of $\alpha$ either contains a complete leaf, or is entirely timelike.
\end{enumerate} 
\end{defn}

Thus, a regular Q-screen contains at least one complete leaf with definite sign of $\alpha$. By shifting $r$  we can take this leaf to be at $r=0$. Moreover, the second generic condition implies that if a screen contains any point $p$ with $\alpha=0$, then an open neighborhood of $p$ contains points with both $\alpha>0$ and $\alpha<0$. Note that this is indeed generic.

\subsection{Past and Future Q-Screens}

Given a quantum marginal surface $\sigma$, we now consider the quantum expansion $\Theta_l$ in the opposite null direction. For example, if $\sigma$ is quantum marginal in the future-outgoing direction, $\Theta_k=0$, we consider the future-ingoing light-rays orthogonal to $\sigma$, with quantum expansion $\Theta_l$. In general $\Theta_l$ need not have uniform sign everywhere on $\sigma$. 

\begin{defn} If $\Theta_l<0$ $(\Theta_l>0)$ everywhere on the quantum marginal surface $\sigma$, we call $\sigma$ {\em marginally quantum (anti)trapped}. \end{defn}

\begin{defn} A {\em future Q-screen} $H$ is a smooth hypersurface admitting a foliation by marginally quantum trapped surfaces called {\em leaves}.  Similarly, a {\em past Q-screen} is a smooth hypersurface foliated by marginally quantum antitrapped leaves. \end{defn}

Recall that $\alpha$ has definite sign on the leaf $\sigma(0)$ of a regular Q-screen, by our earlier convention. For a future (past) Q-screen we shall use the additional convention that $\alpha<0$ ($\alpha>0$) on $\sigma(0)$, which can be implemented by setting $r\to -r$ as needed.

\subsection{New Generalized Second Law}
\label{sec-newgsl}

\begin{conj} Let $H$ be a regular past or future Q-screen, with foliation $\sigma(r)$. Then the generalized entropy $S_\mathrm{gen}(r) \equiv S_\mathrm{gen}[\sigma(r)]$ strictly increases along the foliation:
\begin{equation}
\frac{d S_\mathrm{gen}}{dr} >0~.
\label{eq-gsl}
\end{equation}
\label{conj}
\end{conj}

In fact, we conjecture more strongly that the following geometric properties are obeyed by any regular Q-screen (which need not be past or future):
\begin{itemize} 
\item $\alpha$ cannot change sign anywhere on $H$.
\item The null hypersurface generated by the $k^a$ congruence orthogonal to any leaf $\sigma$ intersects $H$ only on $\sigma$.
\end{itemize} 
As we shall see in Sec.~\ref{sec-proof} (see Remark~\ref{final} and Corollary~\ref{coro}), these two statements are equivalent, and each implies the GSL for past and future Q-screens, Conjecture~\ref{conj}. 

\section{Examples}
\label{sec-examples}

In this section, we consider two examples: one future Q-screen and one past Q-screen. We verify explicitly that they satisfy the new GSL we have proposed.

\subsection{Evaporating Black Hole}
\label{sec-bh}

Consider a Schwarzschild black hole formed by the collapse of a spherically symmetric dust cloud. Fig.~\ref{fig-fqhs} shows the resulting geometry, both with and without Hawking radiation taken into account.\footnote{In the case without evaporation, we take a small, exponentially descreasing density of dust to fall in at all times so as to satisfy the classical generic condition of~\cite{BouEng15a}. In the evaporating case, the quantum generic condition can be satisfied simply by not including any infalling matter at late times.}
It is instructive to study first the classical holographic screen, in {\em both} of these spacetimes. We will then turn to the Q-screen and verify the new GSL. 

\begin{figure*}[ht]
\subfigure[ ]{
\includegraphics[height=0.4 \textwidth]{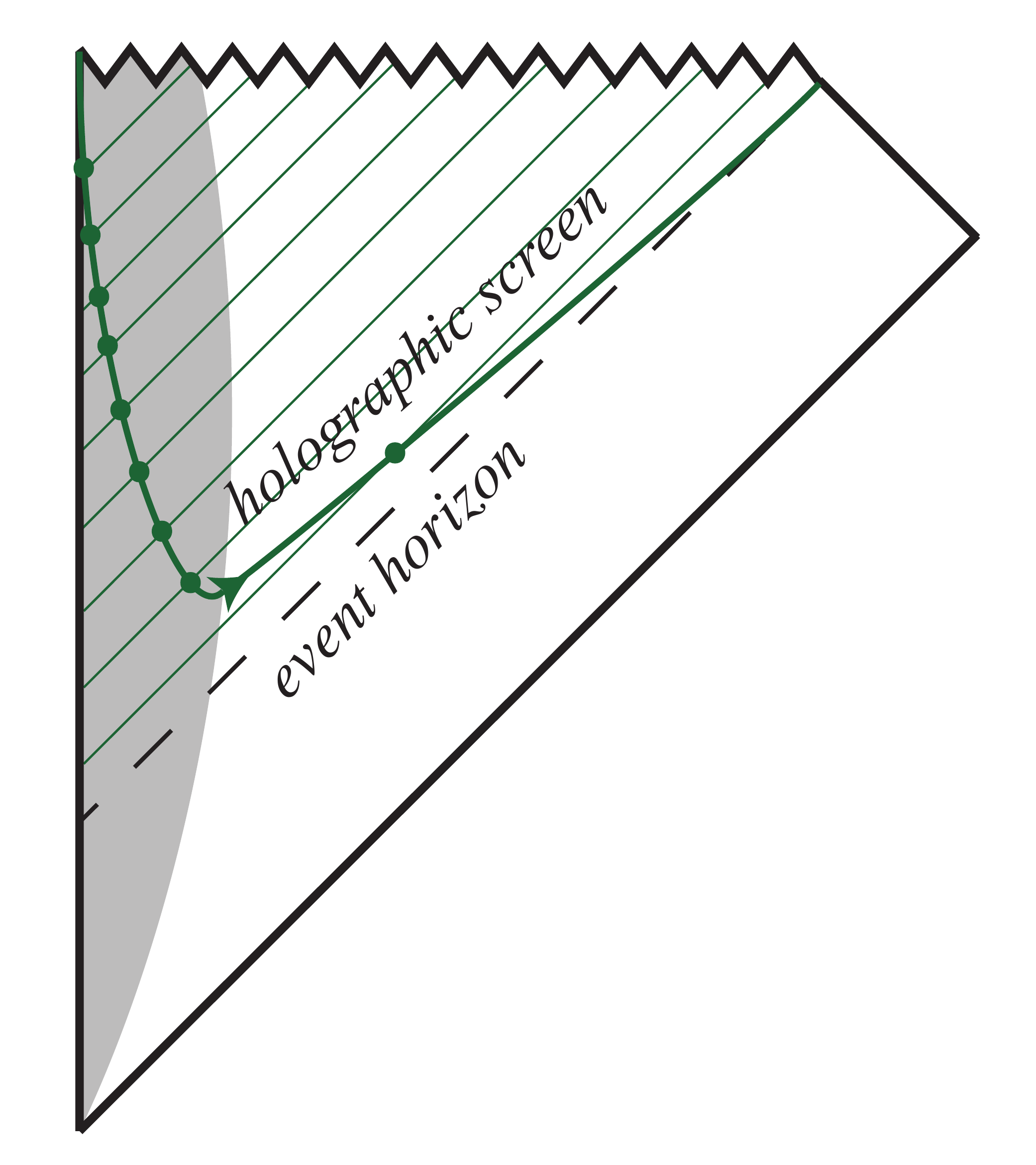}
\label{fig-fqhs-a}}
%\qquad
\subfigure[]{
\includegraphics[height=0.45 \textwidth]{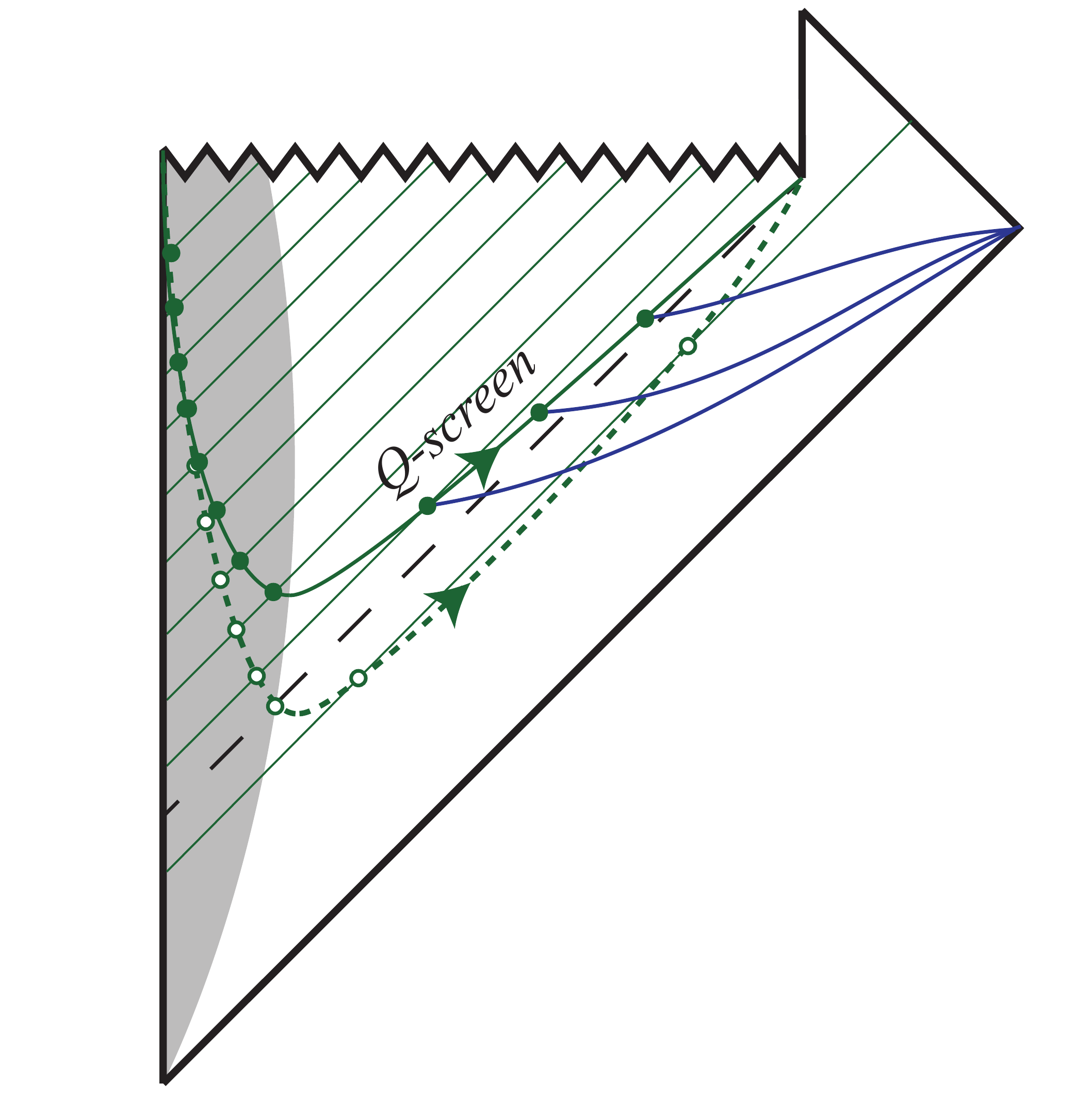}
\label{fig-fqhs-b}}
\caption{Black hole formed by dust collapse. The thin green lines are future light cones which form a null foliation of the spacetime. (a) No Hawking radiation. A dot indicates the marginal surface on each light cone. The area of the classical holographic screen increases towards the exterior and past (arrow). (b) Hawking radiation included. A solid (hollow) dot marks the quantum marginal (marginal) surface(s) on each light cone. The classical screen (short dashed) now lies outside the event horizon (long dashed) during evaporation. A future Q-screen lies inside the black hole. Its area decreases during evaporation. But due to the production of Hawking radiation, the generalized entropy outside the Q-screen increases monotonically, as demanded by our conjecture.}
\label{fig-fqhs}
\end{figure*}

Let $\rho$ be the radial variable usually called $r$ (reserved here for the screen parameter), such that a sphere with coordinate radius $\rho$ has proper area $4\pi\rho^2$. The event horizon is at $\rho=R$; in the evaporating case, $R$ is time-dependent.

\paragraph*{Classical Holographic Screen} The classical holographic screen is constructed by finding spheres of stationary area ($\theta_k=0$) on each of a sequence of future light cones centered at $\rho=0$ (see Fig.~\ref{fig-fqhs-a}). In a classical black hole geometry without evaporation, the screen is contained entirely inside the black hole, because outside the black hole the area of outgoing future light cones grows without bound. (More generally, this follows because the first generic condition of~\cite{BouEng15b} guarantees the existence of a trapped sphere near the sphere of maximal area; and this implies a singularity further along the light cone by Penrose's theorem~\cite{Pen65}.) 

This is a future holographic screen: the expansion in the non-marginal direction is strictly negative everywhere on each of its leaves. Our classical results~\cite{BouEng15b} imply that the screen evolves everywhere to its own past or exterior and that the area grows monotonically under this evolution. In Fig.~\ref{fig-fqhs-a} one can verify this behavior.

Now consider a different geometry that includes backreaction from the Hawking radiation, shown in Fig.~\ref{fig-fqhs-b}. The event horizon grows during the collapse and then shrinks during evaporation. By continuity, sufficiently nearby future light cones just inside or outside the black hole, too, will have a surface that locally maximizes the area. Therefore, the classical holographic screen will extend outside of the event horizon. Now consider a future light cone just barely outside the black hole horizon. Its area grows until it is focused by the collapsing matter; then it shrinks along with the event horizon during a phase when they are formally less than one Planck distance apart. But any light cone that lies outside the horizon will get out to future null infinity, where the area diverges. Therefore the light cone area must have a local minimum during evaporation; this happens when $\rho$ satisfies
\begin{equation}
R(\rho-R)\sim O(l_P^2)~. 
\label{eq-stretch}
\end{equation} 
This is the coordinate radius at which the area of the outgoing light cone would classically increase by about one Planck area per Schwarzschild time, compensating the effect of evaporation. In a typical infalling observer's reference frame, the sphere satisfying Eq.~(\ref{eq-stretch}) has a proper distance of order $l_P$ from the event horizon; thus, the classical screen coincides with the ``stretched horizon''~\cite{SusTho93b} during the evaporation phase.

Hence, each of these barely-exterior light cones contributes two leaves to the classical holographic screen. The behavior of the area in the evaporating phase can be understood as follows. A black hole emits $O(1)$ quanta of energy $T_H\sim \hbar/R$ per Schwarzschild time $R$. This decreases the black hole mass by $O(T_H)$, so the event horizon radius decreases by $O(GT_H)\sim O(l_P^2/R)$. Thus, the area of the event horizon decreases by about one Planck area in every Schwarzschild time $R$. This implies that the area of the minimum sphere on the future light cones just outside the black hole decreases as well. But these are the leaves of the classical holographic screen decreases during the evaporation phase. 

We conclude that the area of the classical holographic screen increases during the collapse phase and decreases during evaporation, when it evolves back to its own future and interior. Though it is a future holographic screen, it does not satisfy our area theorem. This is as expected, much as the event horizon fails to satisfy Hawking's area theorem in this setting, since the NEC is violated.

\paragraph*{Q-Screen} Bekenstein's GSL improves on Hawking's area theorem for event horizons. The matter entropy produced outside the black hole is larger, by a factor $O(1)>1$, than the loss of Bekenstein-Hawking entropy due to the decrease in event horizon area~\cite{Pag76,Pag13}: 
\begin{equation}
\frac{d S_\mathrm{out}}{(-dA/4G\hbar)} -1 \sim O(1) > 0~. 
\label{eq-page}
\end{equation}
Therefore 
\begin{equation}
d S_\mathrm{gen} \equiv \frac{dA}{4G\hbar} + dS_\mathrm{out}  > 0
\label{eq-page2}
\end{equation}
during evaporation, and Bekenstein's GSL is satisfied.

Similarly, we conjectured a quantum improvement of our area theorem: the GSL for Q-screens, Eq.~(\ref{eq-gsl}). 
We will now verify that the conjecture is satisfied in the example of the evaporating black hole. For definiteness, we will chose $\Sigma_\mathrm{out}$ to be the exterior, i.e., the side with the asymptotic boundary. The analysis is unchanged with the opposite choice.

To construct the Q-screen, we again consider outgoing future light cones centered at $\rho=0$, but now we must maximize the generalized entropy along each cone. Outside of the black hole no future cone contains such a maximum. This is obvious for light cones far from the black hole, whose area increases rapidly. Sufficiently close to the event horizon, light cones will decrease in area while they remain less than a Planck distance from the horizon, as discussed above. In this regime the difference with the event horizon is negligible, and it follows from Eq.~(\ref{eq-page2}) that the generalized entropy increases despite the area decrease. Hence the entire Q-screen lies inside the black hole. (More generally, this property follows from Wall's Quantum Singularity Theorem for quantum trapped surfaces~\cite{Wal10QST}.)

On the other hand, every future light cone inside the black hole contains a marginally quantum trapped sphere. On this sphere, the classical decrease in area in the future-outward direction precisely compensates the production of Hawking radiation entropy. By Eq.~(\ref{eq-page}), this implies that at the marginally quantum trapped sphere, the light cone's area must shrink faster than the event horizon (since the latter's decrease does not fully compensate the radiation entropy)\footnote{More precisely, the light cone area must decrease by $O(l_P^2)$ as $v\to v+R$, where $v=t + \rho + R\log\left|\frac{\rho}{R}-1\right|$.}. By the same reasoning that led to Eq.~(\ref{eq-stretch}), this will occur where $\rho$ satisfies
\begin{equation}
R(\rho-R)\sim -O(l_P^2)~,
\label{eq-shrink}
\end{equation} 
corresponding to a proper distance of order $l_P$ inside the event horizon, as measured by an infalling observer.

In summary, we find that the Q-screen is like a ``shrunk horizon'': during the evaporation phase, it hovers a Planck distance inside the event horizon. Thus, the Q-screen partakes in the event horizon's decrease; its area is always of order a Planck area smaller than the event horizon. It is now clear that our GSL is obeyed, for the same reason that Bekenstein's GSL is obeyed: we know from Eq.~(\ref{eq-page2}) that the area decrease along the Q-screen is more than compensated by the production of Hawking radiation entropy in the black hole's exterior.

\subsection{Cosmology}
\label{sec-cosmo}

Consider a flat, expanding Friedmann-Robertson-Walker cosmology filled with radiation. The metric is
\begin{equation}
ds^2 = -dt^2 + a(t)^2 [d\chi^2 + \chi^2 d\Omega^2]~,
\end{equation}
with $a(t) = (t/l_P)^{1/2}$. Again we will construct both a classical screen and a Q-screen. We will find that both are past screens, and we will verify our area theorem and our new GSL.
\begin{figure}[ht]
\includegraphics[height=0.40 \textwidth]{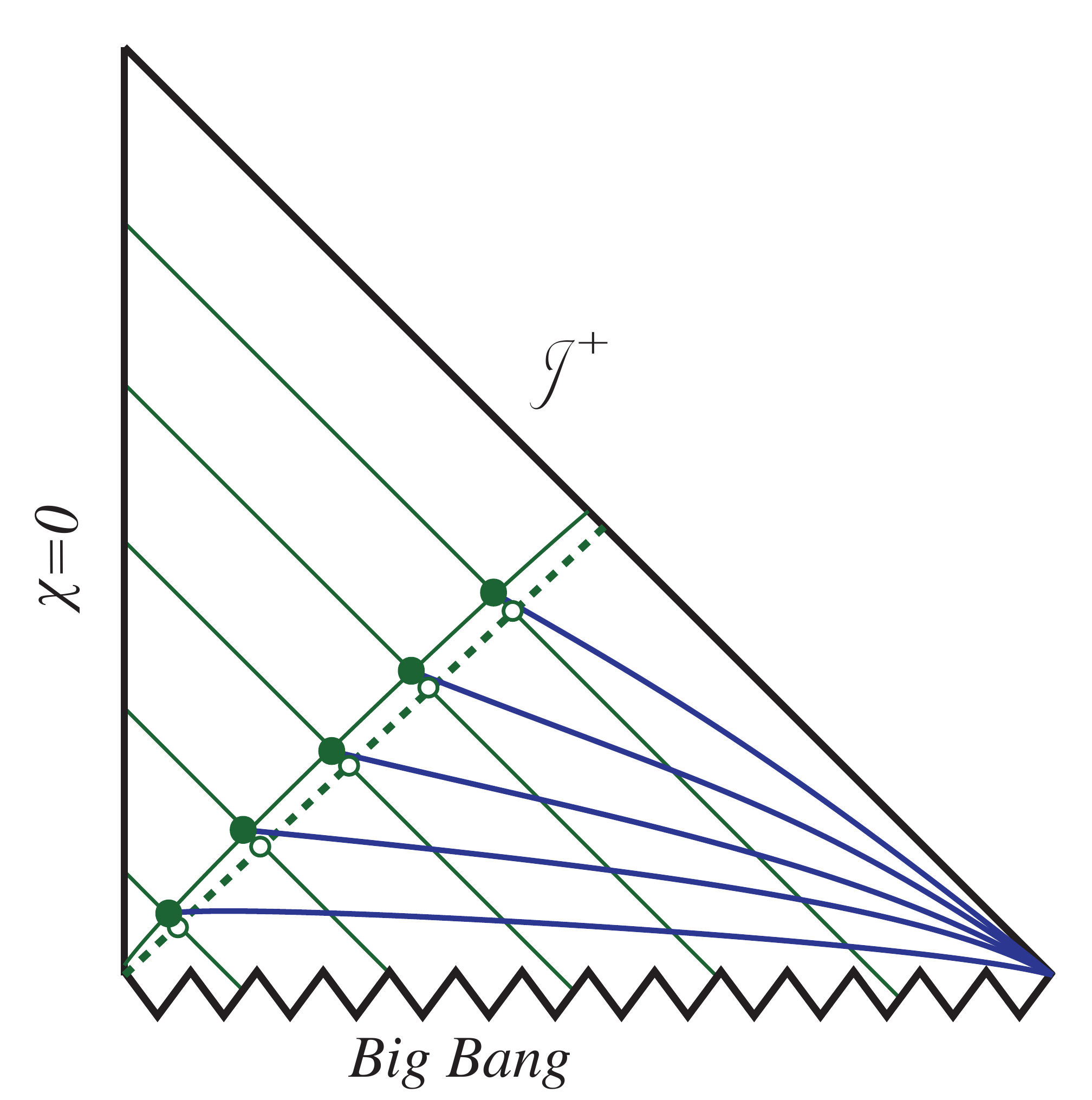}
\caption{Radiation dominated expanding universe; dots and lines as in Fig.~2b. The classical and Q-screen nearly coincide; the area and generalized entropy both grow monotonically to the future.}
\label{fig-pqhs}
\end{figure}

We begin by picking a simple null foliation of the spacetime: the past light cones of an ``observer'' at $\chi=0$, see Fig.~\ref{fig-pqhs}. A classical screen is constructed by maximizing the area on each cone and combining the corresponding spheres into a hypersurface.
%The energy density of the radiation is
%\begin{equation}
%\rho(t) = \frac{3}{32\pi G t^2}~.
%\end{equation}
Consider a past light cone with tip at the time corresponding to scale factor $a_0$. The area of the sphere at scale factor $a$ is
\begin{equation}
A(a) = 16\pi a^2(a_0-a)^2 l_P^2~.
\end{equation}
This is maximal at $a=a_\mathrm{PHS} = a_0/2$ or $t_\mathrm{PHS} = t_0/4$. Thus we find that the area of the classical screen increases monotonically towards the future,
\begin{equation}
\frac{dA_\mathrm{PHS}}{d t_\mathrm{PHS}} = 32\pi t_\mathrm{PHS} >0 ~,
\end{equation}
as guaranteed by our area theorem~\cite{BouEng15a}. 

To construct the quantum screen, we use the same null foliation, but now we maximize the generalized entropy on each past light cone:
\begin{equation}
d\left[A+ 4G\hbar\, S_\mathrm{out}\right]=0
\label{eq-balance}
\end{equation}
Assuming the number of massless species to be of order unity, the entropy per comoving volume is $s\sim l_P^{-3}$, and we have
\begin{equation}
\frac{d S_\mathrm{out} }{dr} = 4\pi \chi^2\, s
\end{equation}
If we choose $\Sigma_\mathrm{out}$ to be the exterior\footnote{The exterior entropy diverges since the volume is infinite. It can be regulated by taking the edge of $\Sigma_\mathrm{out}$ to lie near the big bang at some large but finite comoving radius. The edge is held fixed as $\sigma$ is varied.} (interior) of the past light cone, $S_\mathrm{out}$ will increase (decrease) monotonically as we move to the future on the light cone. By Eq.~(\ref{eq-balance}), this implies that the maximum of the generalized entropy will not be exactly in the same place as the maximum of the area on the same light cone. Instead, it will be shifted slightly inward and to the future (outward and to the past). The shift is of order the geometric mean of the age of the universe and the Planck time:
\begin{equation}
|t_\mathrm{QS} - t_\mathrm{PHS}| \sim O(\sqrt{t_\mathrm{PHS} l_P})~.
\end{equation}
The behavior of generalized entropy along the Q-screen is dominated by the classical growth of the area. We find
\begin{equation}
l_{\rm P}^2~ \frac{d S_\mathrm{gen,QS}}{d t_\mathrm{QS}} = 32\pi t_\mathrm{QS} \pm O(\sqrt{t_\mathrm{QS} l_P}) > 0~,
\end{equation}
consistent with our conjectured GSL in the semiclassical regime, $t_\mathrm{QS}\gg l_P$.

\section{Proof from the Quantum Focussing Conjecture}
\label{sec-proof}

In this section, we show that the Generalized Second Law for Q-screens follows from the Quantum Focussing Conjecture (QFC)~\cite{BouFis15a}. 

Why derive one conjecture from another? The first reason is that it is useful to understand the logical structure of a set of plausible and interesting conjectures. Our result establishes that the QFC is at least as strong as the new GSL. However, the new GSL is not obvious from the QFC: the implication requires a nontrivial proof.

Secondly, in light of the proof below, any evidence that makes the QFC more plausible can be regarded in particular as evidence for the new GSL. Indeed, there is considerable evidence for the QFC: it implies several nontrivial related statements which have already been proven or extensively tested. In the classical limit of the geometry and the stress tensor, the QFC implies the classical focussing property of General Relativity. For null hyperplanes in Minkowski space, the QFC implies a novel lower bound on the quantum stress tensor in terms of the second derivative of the exterior entropy. This ``Quantum Null Energy Condition'' was recently proven~\cite{BouFis15b}. Finally, the QFC also implies the Bousso bound~\cite{CEB1} on the entropy crossing a lightsheet~\cite{BouFis15a}. No counterexample to this bound is known, and the bound has been proven in certain hydrodynamic regimes~\cite{FMW,BouFla03}.

\subsection{Quantum Focussing Conjecture}
\label{sec-qfc}

The Quantum Focussing Conjecture (QFC) states that the quantum expansion cannot increase along any null congruence~\cite{BouFis15a}. More precisely,
\begin{equation}
\frac{\delta}{\delta V(y_2)} \Theta_k[V(y); y_1] \le 0~.
\label{eq-sqfc}
\end{equation}
The quantum expansion $\Theta$ is defined as in Sec.~\ref{sec-conjecture}, except that we characterize the surface $\sigma$ that appears in Eq.~(\ref{eq-qexpdef}) in terms of its affine position $V(y)$ on some null hypersurface $N$ generated by a congruence of null geodesics $y$, with tangent vector $k^a$. Thus, the QFC states that the quantum expansion cannot increase at $y_1$, if $\sigma$ is infinitesimally deformed along the generator $y_2$ of $N$, in the $k^a$ direction. Here $y_2$ can be taken to be either the same or different from $y_1$. 

Suppose that the quantum expansion in the orthogonal null direction $k^a$ is nonpositive (negative) somewhere on $\sigma$, i.e., suppose that $\Theta_k[V(y);y_1]\leq ~0$ for some $y_1$. Then $\Theta(\nu;y_1)$ will remain nonpositive at the null geodesic $y_1$, under forward evolution of the surface in the $k^a$ direction. More precisely, for two slices of $N$ satisfying $V'(y)\geq V(y)$ (for all $y$), the QFC implies that
\begin{equation}
\Theta[V(y), y_1] \le 0~,~ V(y)\geq 0 \implies \Theta[V(y),y_1] \le 0~,
\label{eq-wqfc}
\end{equation}
where if the first inequality is strict, then so is the second (unless $V'$ and $V$ coincide for all $y$).

In particular, if the expansion is negative everywhere on some surface $\sigma$, then it cannot vanish on any cross-section of $N$ that lies entirely in the $k^a$ direction away from $\sigma$. This will be the specific consequence of the QFC that enters the proof below. The statement is analogous to the classical result in General Relativity, that if light-rays are converging in a spacetime satisfying the NEC, then they cannot begin to diverge at any regular point of the congruence. 

\subsection{Derivation of the New Generalized Second Law}
\label{sec-derivation}

The derivation of the new GSL from the QFC will be closely analogous to our proof of the area law for holographic screens from the Null Energy Condition. Roughly, we will replace the classical expansion $\theta$ with the quantum expansion $\Theta$, and the assumption of the Null Energy Condition with the assumption of the QFC. 
 
We begin by recalling an important set of definitions and results from Ref.~\cite{BouEng15b}, which are purely geometric and carry over unchanged. Any Cauchy-splitting surface $\sigma$ defines a partition of the spacetime into sets $K^\pm(\sigma)$ whose shared boundary is a null hypersurface $N(\sigma)$ orthogonal to $\sigma$; see Fig.~\ref{fig-sets}. 
\begin{figure*}[ht]
\subfigure[ ]{
\includegraphics[height=0.28 \textwidth]{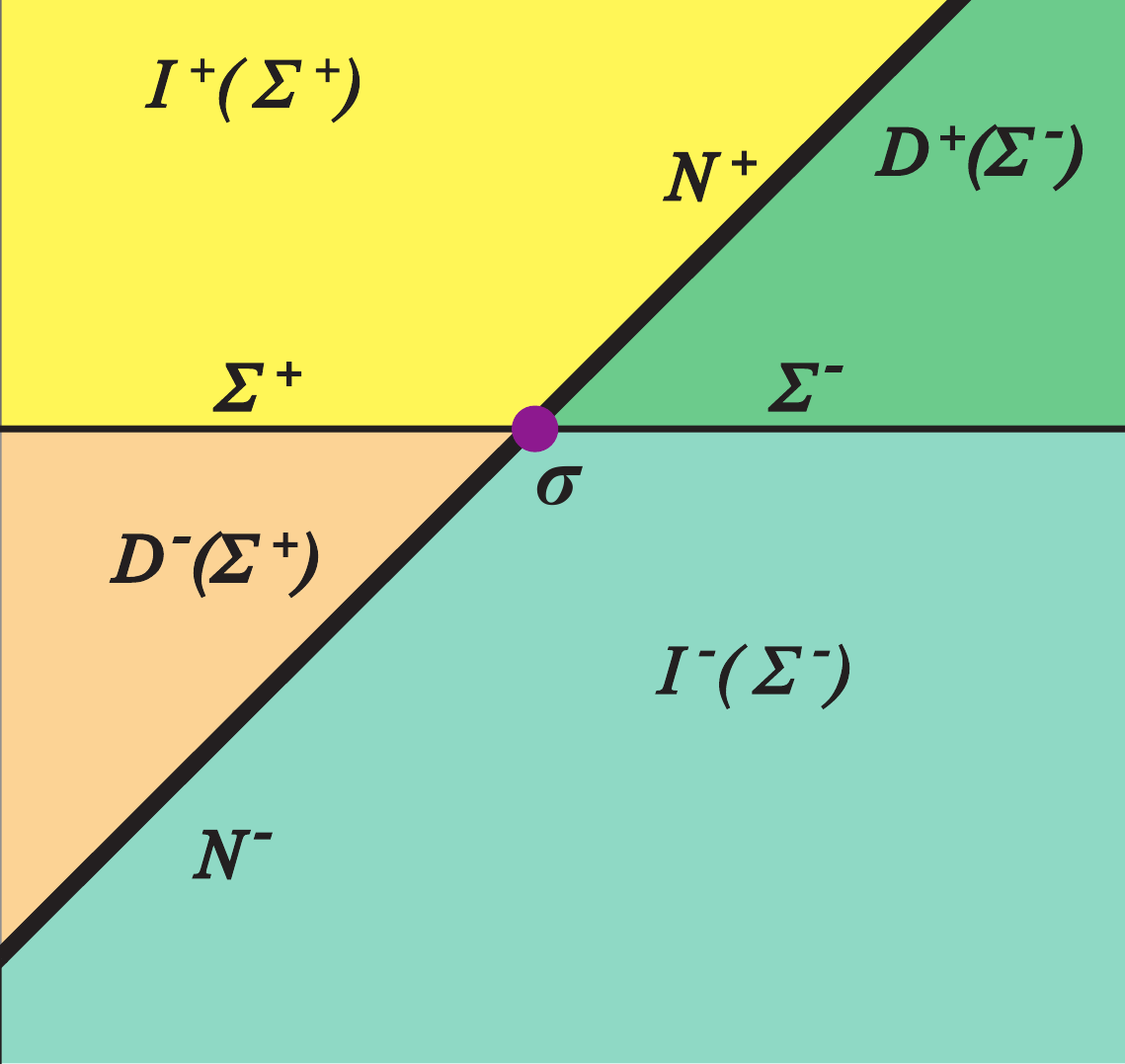}
\label{fig-sets-a}}
%\qquad
\subfigure[]{
\includegraphics[height=0.28\textwidth]{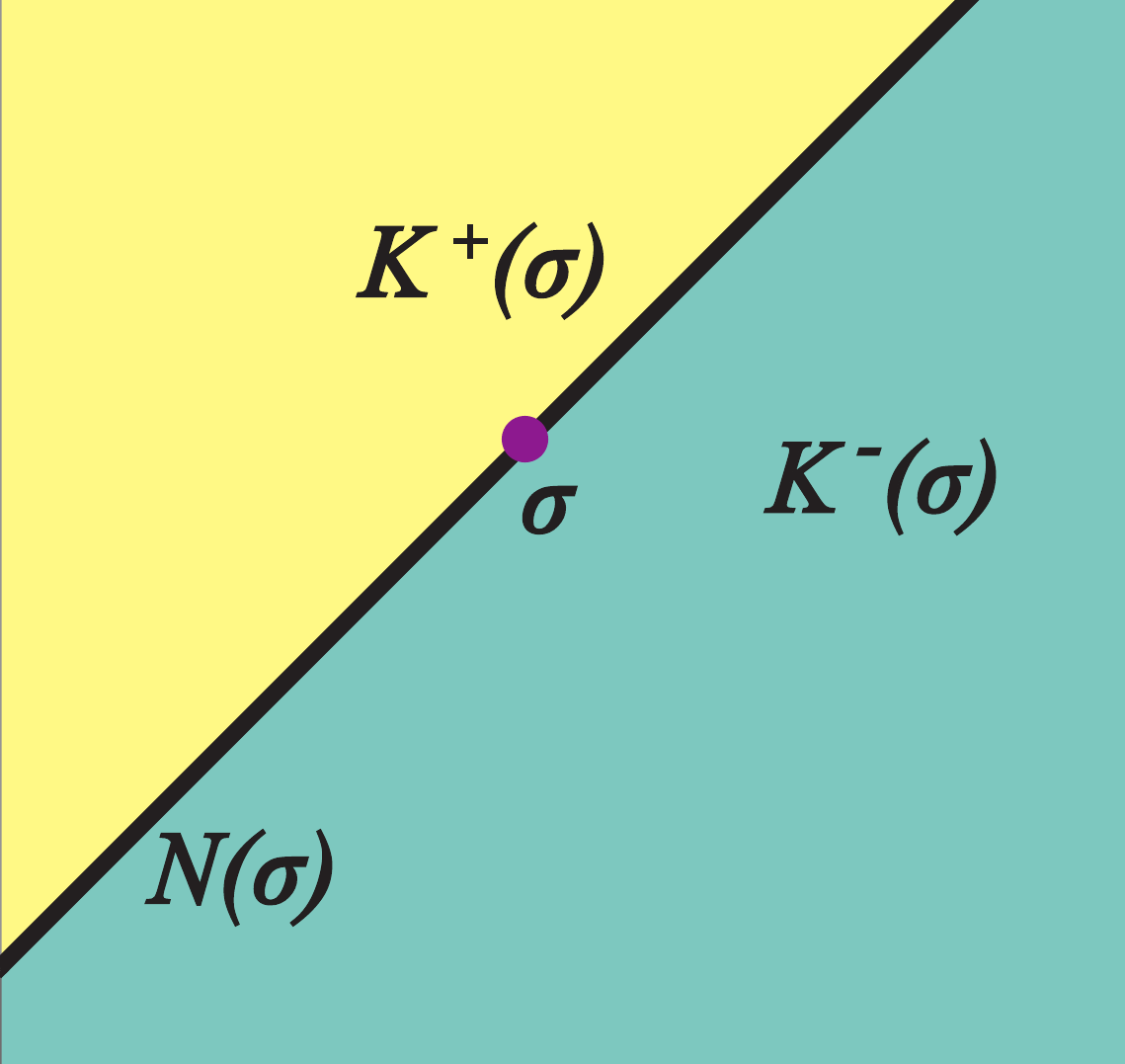}
\label{fig-sets-b}}
%\qquad
\subfigure[]{
\includegraphics[height=0.28\textwidth]{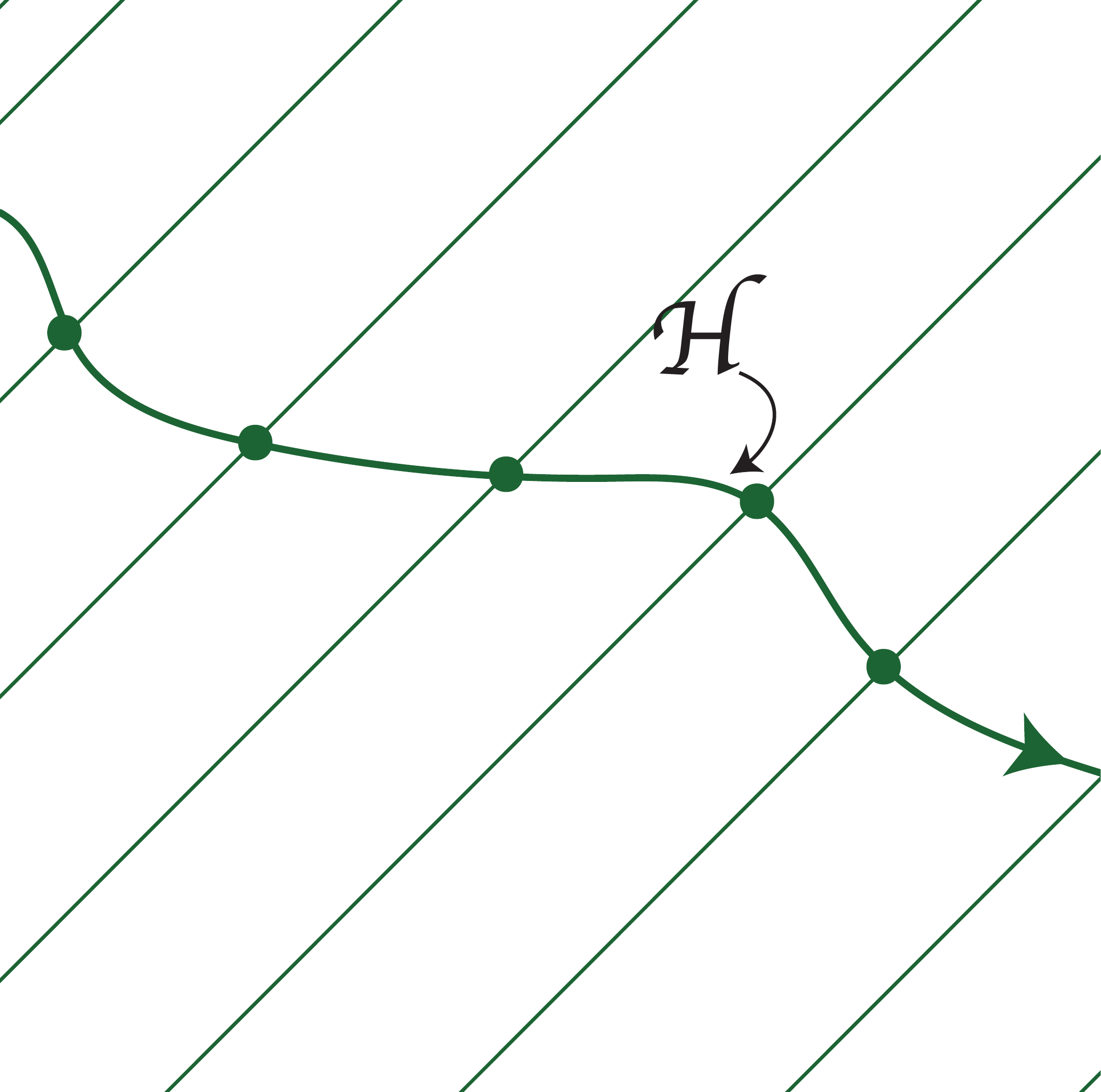}
\label{fig-sets-c}}
\caption{(a) A surface $\sigma$ that splits a Cauchy surface defines a partition of the entire spacetime into four regions, given by the past or future domains of dependence and the chronological future or past of the two partial Cauchy surfaces $\Sigma^\pm$. (b) The pairwise unions $K^\pm$ depend only on $\sigma$, not on the choice of Cauchy surface. $K^\pm$ share a boundary $N = N^+\cup N^- \cup \sigma$ generated by light-rays orthogonal to $\sigma$. (c) If a hypersurface $\cal H$ foliated by $\sigma(r)$ has $\alpha<0$ everywhere (see text for definition), then the sets $K^+(r)$ are monotonic under inclusion, and the sets $N(r)$ define a null foliation of spacetime.}
\label{fig-sets}
\end{figure*}
Now consider any hypersurface $\cal H$ foliated by surfaces $\sigma(r)$, with tangent vector field $h^a = \alpha l^a + \beta k^a$ normal to the foliation, as defined in Eq.~(\ref{eq-hlk}). Note that the surfaces $\sigma(r)$ need not be marginal; we use the notation $\cal H$ rather than $H$ for a hypersurface of this more general type. In~\cite{BouEng15b} we proved that if $\alpha$ has definite sign on $\cal H$, then the sets $K^\pm(r)\equiv K^\pm(\sigma(r))$ are monotonic under inclusion:

\begin{lem} Let $r_1<r_2$. If $\alpha<0$ everywhere on ${\cal H}$, then $\bar K^+(r_1)\subset K^+(r_2)$ and $K^-(r_1)\supset \bar K^-(r_2)$, where an overbar denotes closure. If $\alpha>0$ everywhere on $\cal H$, then $\bar K^+(r_2)\subset K^+(r_1)$ and $K^-(r_2)\supset \bar K^-(r_1)$.
\label{lem-kmono}
\end{lem}

We will also need an important result due to Wall, which constrains the quantum expansion of a surface that touches but does not cross a null hypersurface $N$. Let $\chi$ be a spacelike surface tangent to $N$ at a point $p$. That is, we assume that one of the two future-directed null vectors orthogonal to $\chi$, $\kappa^a$, is also orthogonal to $N$ at $p$. We may normalize the (null) normal vector field to $N$ so that it coincides with $\kappa^a$ at $p$.  

We further assume that both $\chi$, and $\nu \equiv N\cap \Sigma$, split a Cauchy surface $\Sigma$. We pick an arbitrary side $\Sigma_\mathrm{out}(\chi)$ as the ``outside'' of $\chi$, and we choose $\Sigma_\mathrm{out}(\nu)$ to be the ``same'' side, in the sense that they agree at $p$. This defines generalized entropies $S_\mathrm{gen}[\chi]$, $S_\mathrm{gen}[\nu]$. We define the quantum expansions $\Theta[\chi;y]$, $\Theta[\nu;z]$ with respect to that null vector field orthogonal to each surface which coincides with $\kappa$ at $p$. 

Let $y$ be Gaussian normal coordinates about $p$ on $\chi$. Because $\chi$ and $\nu$ are tangent at $p$, they can be identified at linear order in the distance $\delta$ from $p$. Thus, in a sufficiently small neighborhood of $p$, we can use the same coordinates $z=y$ on $\nu$, up to an $O(\delta^2)$ ambiguity which will be irrelevant.

\begin{lem}  \label{lem:aron}
Let $\chi$ and $\nu$ be Cauchy-splitting surfaces tangent at a point $p$, and let $N\supset \nu$ be a null hypersurface, as described above.
\begin{itemize} 
\item If $\chi$ lies entirely outside the past of $N$, then any small open neighborhood of $p$ contains a point $y$ such that $\Theta[\chi; y]\geq\Theta[\nu;y]$.
\item If $\chi$  lies entirely outside the future of $N$, then any small open neighborhood of $p$ contains a point $y$ such that $\Theta[\chi; y]\leq\Theta[\nu;y]$.
\end{itemize} 
\end{lem}

\begin{proof}
By causality, the entire null hypersurface $N(\chi)$, defined as the boundary of $K^+(\chi)$ (Fig.~\ref{fig-sets-b}), is nowhere to the past of $N$ in the first case, and nowhere to the future of $N$ in the second case. If the outside is chosen to be the side to which $\kappa^a$ points, then the first claim is identical to Theorem 1 in~\cite{Wal10QST}, and the second claim follows by exchanging $N$ with $N(\chi)$. With the opposite choice of exterior, the proof can be reduced to the above cases by time reversal.
\end{proof}

The proof of Conjecture~\ref{conj} now proceeds in two steps. First we will combine the monotonicity property of $K^\pm$ with the QFC to show that $\alpha$ {\em must} have definite sign on a Q-screen $H$. Then we show that this implies the new GSL, Eq.~(\ref{eq-gsl}), if in addition $H$ is a past or future Q-screen. 

\begin{thm}\label{thm:structure} Let $H$ be a regular Q-screen in a spacetime satisfying the QFC, and let $\alpha$ be defined by Eq.~(\ref{eq-hlk}). Then $\alpha$ has definite sign on $H$. That is, either $\alpha<0$ everywhere on $H$, or $\alpha>0$ everywhere on $H$.
\label{theorem}
\end{thm}

\begin{proof} 

%We will use the same notation as~\cite{} to denote spacelike and timelike portions of the holographic screen: a (spacelike) region is denoted $S_{-+}$ if on it, $\alpha<0, \beta>0$; a timelike region denoted $T_{--}$ if $\alpha<0, \beta<0$; $S_{+-}$ if $\alpha>0, \beta<0$; and finally, $T_{++}$ if $\alpha>0, \beta>0$. Because $\alpha$ and $\beta$ cannot simultaneously vanish, by definition of a regular screen, \clas{ $H$ contains at least one complete leaf with definite sign of $\alpha$.} Let $\sigma(0)$ be such a leaf. \clas{ We will take the parameter $r$ to be oriented so that $\alpha< 0$ on $\sigma(0)$, and we take $r=0$ on $\sigma(0)$. This also determines the global orientation of the vector field $h^a$.}\footnote{\clas{For past holographic screens, it is convenient to choose the opposite convention, $\alpha>0$ on $\sigma(0)$.}}

By the condition~\ref{def-technical}.c and the subsequent convention, $\alpha$ has definite sign on the leaf $\sigma(0)$. If $\alpha>0$ at $r=0$, we can reparametrize $r\to -r$, so without loss of generality we may assume that $\alpha<0$ at $r=0$. We will now show that $\alpha<0$ everywhere on $H$.

Suppose for contradiction that $H$ contains a point with $\alpha\geq 0$. Then the subset $H_+\subset H$ of points with $\alpha> 0$ is also nonempty, by Assumption~\ref{def-technical}.b. Continuity guarantees that $\alpha<0$ in an open neighborhood of the leaf $\sigma(0)$, so $H_+$ has a connected component entirely in the $r>0$ region, or entirely in the $r<0$ region (or both). We first consider the case $r>0$.

It is convenient to rescale $r$ to set
\begin{equation}
1=\inf\{r:r>0,\sigma(r)\cap H_+ \neq \varnothing\}~.
\end{equation}
Then by the second generic condition~\ref{def-technical}.b, $\alpha<0$ for all leaves $\sigma(r)$ with $0<r<1$. Hence by Lemma~\ref{lem-kmono}, there exists an open neighborhood of $\bar K^-(1)$ that is contained in $K^-(0)$, and for sufficiently small $\epsilon$ we have
\begin{equation}
K^-(0)\supset K^-(1+\epsilon)~.
\label{eq-nocon}
\end{equation}

By continuity, the set $P$ of points on $\sigma(1)$ with $\alpha=0$ is nonempty. $P$ may consist of several connected components $P_i$. We cannot assume that $\beta$ is of fixed sign for $0<r<1$. But since $\alpha$ and $\beta$ cannot vanish simultaneously, $\beta$ has fixed sign in an open neighborhood $O(P_i)$ of each $P_i$. However, $\beta$ need not have the same sign in all of these neighborhoods. We distinguish two complementary cases.

\paragraph*{Case 1} We first consider the case where $\beta>0$ in every $O(P_i)$.  Then the assumed sign change from $\alpha<0$ to $\alpha>0$ corresponds to a transition of $h^a$ from spacelike-outward ($S_{-+}$) to timelike-future-directed ($T_{++}$).  

Let $\sigma^+(1+\epsilon)$ be the set of points with $\alpha>0$ on the leaf $\sigma(1+\epsilon)$. Note that $\sigma^{+}(1+\epsilon)$ may be disconnected, but each disconnected component is open. 

By choosing $\epsilon$ sufficiently small, we can ensure that each connected component of $\sigma^+(1+\epsilon)$ is contained in a single neighborhood $O(P_i)$. Let $\Gamma$ be the set of integral curves of $h^a$ that pass through $\sigma^+(1+\epsilon)$. Note that each such curve can also be parametrized by $r$. 

Because $\alpha>0$, each curve in $\Gamma$ lies in $K^-(1+\epsilon)$ in some range $r_\phi< r<1+\epsilon$. By Eq.~(\ref{eq-nocon}), $\sigma(0)\cap K^-(1+\epsilon)=\varnothing$, so $r_\phi>0$. At $r_\phi$, the curve intersects the boundary $N(1+\epsilon)$ of $K^-(1+\epsilon)$. Because $\beta>0$ in $O(p)$, this intersection will be with $N^-(1+\epsilon)$. By smoothness and the second generic assumption, the intersection will consist of one point per curve, $r=r_\phi$.

Let the spatial surface $\phi$ be the set of points $r=r_\phi$ of the curves in $\Gamma$. The sets $\phi$ and $\sigma^+(1+\epsilon)$ have the same topology because the integral curves define a continuous, one-to-one map between them. The closures of both sets, $\bar\sigma^+(1+\epsilon)$ and $\bar\phi$, are also related by this map and share a boundary at $r=1+\epsilon$.

Since $\bar\sigma^+(1+\epsilon)$ is a closed subset of a compact set, it is compact; and by the fiber map, $\bar\phi$ is also compact. Therefore the global minimum $R\equiv\inf\{r(p): p\in \phi\}$ is attained on one or more points $Q\subset \bar\phi$. Since $R<1$ but $\dot\phi\subset \sigma(1+\epsilon)$, $Q\notin\dot\phi$, so $Q$ consists of stationary points of $r$, viewed as a function on $\phi$. Hence the leaf $\sigma(R)$ is tangent to the null hypersurface $N^{-}(1+\epsilon)$ at $Q$.

Because $Q$ achieves a global minimum of $r$ on $\bar\phi$, $\sigma(R)$ lies nowhere in the past of $N^{-}(1+\epsilon)$. Let the spacelike surface $\nu\supset Q$ be a compact cross-section of $N(1+\epsilon)$; since $N(1+\epsilon)$ is spacetime-splitting, $\nu$ will be Cauchy splitting. Because $\sigma(R)$ is tangent to $N(1+\epsilon)$ at $Q$, Lemma~\ref{lem-kmono} implies that any open neighborhood of $Q$ contains a point $y$ such that $\Theta_k[\sigma(R);y]\geq\Theta_k[\nu;y]$. By the QFC, $\Theta_k[\nu,y]\geq\Theta_k[\sigma(1+\epsilon),y]=0$; and by the first generic condition, the inequality is strict, so $\Theta_k[\nu,y]>0$. Hence $\Theta_k[\sigma(R);y]>0$. But this contradicts the defining property of a Q-screen, that the quantum expansion of each leaf $\sigma$ in the $k^a$ direction must vanish.

\paragraph*{Case 2} We now consider the case where $\beta<0$ in at least one open neighborhood $O(P_1)$. We showed in~\cite{BouEng15b} that this implies the existence of a transition with $\tilde \beta>0$ elsewhere on $H$, on a leaf $\sigma(2)$, under reversal of the flow direction, $\tilde r\equiv 3-r$. (We use the tilde to denote quantities defined with respect to the reverse flow.) Upon closer inspection, one finds that our argument establishes a stronger result that was not needed in~\cite{BouEng15b}: that $\tilde \beta>0$ on {\em all} neighborhoods $O(\tilde P_j)$ of transition points on $\sigma(2)$. Namely, we showed that the only type of $\alpha<0$ region that can end at $\sigma(2)$ under the original flow is a timelike region, i.e., $\beta = -\tilde \beta <0$; see Fig.~8 of Ref.~\cite{BouEng15b}. This implies a case 1 transition (in the sense of the present paper) on $\sigma(2)$. Since we have already shown that case 1 transitions are impossible, we can now conclude that case 2 transitions are also impossible.

\begin{figure}[t]
\includegraphics[width=.45 \textwidth]{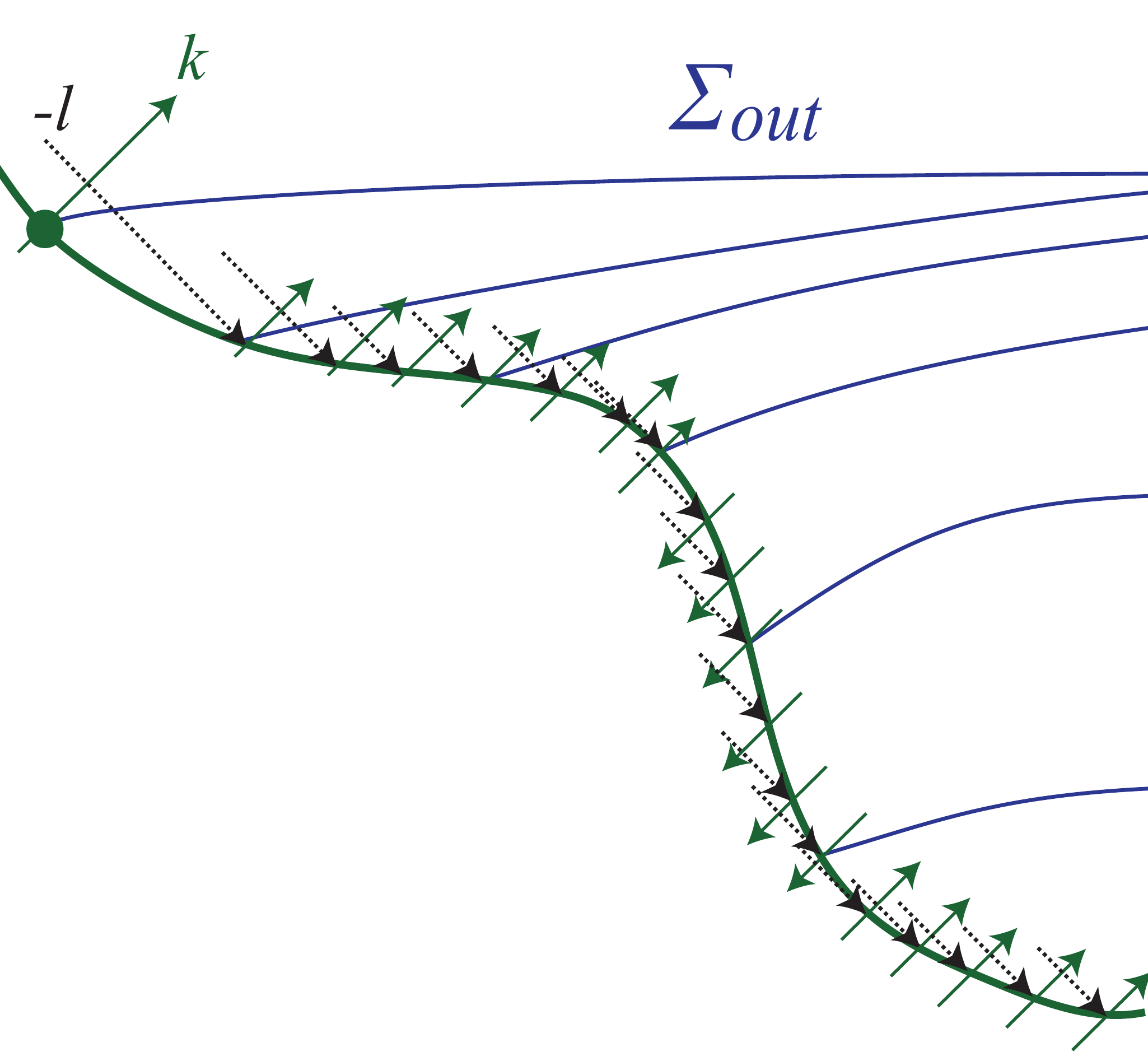}
\caption{The flow from leaf to leaf along a Q-screen can be decomposed as a sequence of infinitesimal motions in the $k^a$ and $l^a$ null directions. In the $\pm k^a$ direction, the generalized entropy is locally stationary by definition of the Q-screen, $\Theta_k=0$. Because $\alpha<0$ by Theorem~\ref{thm:structure}, the motion is always towards $-l^a$, along which the generalized entropy increases since $\Theta_l<0$. Hence the generalized entropy increases along the flow.}
\label{fig-zigzag}
\end{figure}

\paragraph*{Cases 3 and 4} Now suppose that a transition to $\alpha>0$ occurs at some $r<0$. Case 3 arises if $\beta>0$ everywhere at the onset of the transition. Case 4 is the complementary case where $\beta<0$ in at least one connected component. A straightforward adaptation of the case 1 and 2 analyses as in~\cite{BouEng15b} rules out the possibility of case 3 and 4 transitions.

In summary, since $\alpha<0$ at $r=0$ and no transitions to $\alpha>0$ are possible, it follows from the second generic condition~\ref{def-technical2} that $\alpha<0$ everywhere on $H$.
\end{proof}

The flow along $H$ with increasing $r$ can be deformed into a ``zig-zag'' flow along null surfaces orthogonal to the leaves $\sigma(r)$ and $\sigma(r+dr)$; see Fig.~\ref{fig-zigzag}, and see Ref.~\cite{BouEng15a} for further details. Locally the flow will be in the $+k^a$ direction where $H$ is spacelike; it will be in the $-k^a$ direction where $H$ is timelike. But because $\alpha<0$, the flow will always be in the $-l^a$ direction, never in the $+l^a$ direction. That is, the flow towards larger $r$ corresponds to a flow to the exterior or past. 

We now show that Theorem~\ref{theorem} implies Conjecture~\ref{conj}, by applying it to the case where the regular Q-screen $H$ is in addition past or future, as assumed in our conjecture. 
\begin{proof} By the definition of a future Q-screen, each of its leaves is marginally quantum trapped, with $k^a$ being the marginal direction. Thus to first order in $dr$, the generalized entropy does not change in the $k^a$ direction, and it strictly increases in the $-l^a$ direction. This implies the new GSL, Eq.~(\ref{eq-gsl}), for future holographic screens: 
\begin{equation}
\frac{d S_\mathrm{gen}}{dr}>0~.
\end{equation}
Similarly, the new GSL follows for past Q-screens, where the generalized entropy increases in the $+l^a$ direction, i.e., towards the exterior (in spacelike portions of $H$) or the future (in the timelike portions).
\end{proof}

We stress again that the QFC is itself unproven; we have established a logical relation between what we regard as two plausible conjectures. We also obtained Theorem~\ref{theorem} as a key intermediate result. If we do not wish to assume the QFC, then Theorem~\ref{theorem} can still be considered, as the first of the two stronger conjectures we made in Sec.~\ref{sec-newgsl}. 
\begin{rem} 
The above short proof establishes that Theorem~\ref{theorem} (viewed as a conjecture) is indeed stronger than our new GSL, Conjecture~\ref{conj}. 
\label{final}
\end{rem}
In Sec.~\ref{sec-newgsl} we further claimed the following equivalence:
\begin{cor}
Theorem~\ref{theorem} holds if and only if $N(r)$ intersects the regular Q-screen $H$ only on $\sigma(r)$. 
\label{coro}
\end{cor}
\begin{proof}
To prove {\em if}, suppose first that $\alpha$ did change sign on $H$. This was in fact assumed in our proof of Theorem~\ref{theorem}, and it was shown to imply that some $N(r)$ will intersect $H$ at a point that is not contained in $\sigma(r)$. To prove {\em only if}, suppose that there existed some $N(r_1)$ that intersects $H$ at a point $p\in \sigma(r_2)$ with $r_2\neq r_1$. We may assume that $r_2>r_1$ by setting $r\to -r$ as needed. Since $\sigma(r_2)\subset \bar K^\pm(r_2)$, Lemma~\ref{lem-kmono} implies that $p\in K^+(r_1)$ if $\alpha>0$ everywhere, and that $p\in K^-(r_1)$ if $\alpha<0$ everywhere on $H$. But this is impossible since $K^\pm$ are open sets and $p\in N(r_1)= \dot K^\pm(r_1)$. Hence $\alpha$ must change sign on $H$.
\end{proof}

\end{spacing}
\acknowledgments
It is a pleasure to thank Z.~Fisher and A.~Wall for discussions. NE thanks the Berkeley Center for Theoretical Physics and the UC Berkeley Physics Department for their hospitality. The work of RB is supported in part by the Berkeley Center for Theoretical Physics, by the National Science Foundation (award numbers 1214644, 1316783, and 1521446), by fqxi grant RFP3-1323, and by the US Department of Energy under Contract DE-AC02-05CH11231. The work of NE is supported in part by the US NSF Graduate Research Fellowship under Grant No. DGE-1144085, by NSF Grant No. PHY-1504541, and by funds from the University of California.

\bibliographystyle{utcaps}
\bibliography{all}
\end{document}